\documentclass[11pt,twoside]{article}

\usepackage[ansinew]{inputenc} 
\usepackage{amsmath,amsfonts,amssymb,amsthm}
\usepackage{pstricks,pst-node,pst-coil,pst-plot,pstricks-add}
\usepackage{geometry,epsfig}
\usepackage{bbm}
\usepackage{cite}
\usepackage{graphicx}

\usepackage{bbm}
\usepackage{cite}
\usepackage{paralist}
\usepackage{enumitem}
\usepackage{emptypage}
\usepackage[hidelinks]{hyperref}

\DeclareGraphicsExtensions{.pdf}

\usepackage{mathrsfs} 

\usepackage[T1]{fontenc}

\newcommand{\field}[1]{\mathbb{#1}}
\newcommand{\N}{\field{N}}
\newcommand{\R}{\field{R}}
\newcommand{\C}{\field{C}}
\newcommand{\Z}{\field{Z}}

\newcommand{\HH}{\mathscr H}

\newcommand{\LL}{\mathscr L}

\newcommand{\eps}{\varepsilon}
\newcommand{\ph}{\varphi}

\newcommand{\ran}{\mathrm{Ran}}
\renewcommand{\ker}{\mathrm{Ker}}

\newcommand{\norm}[1]{\mbox{$\left\| #1 \right\|$}}   
\newcommand{\sprod}[2]{\mbox{$\langle #1,#2 \rangle$}}   
\newcommand{\ket}[1]{\left| #1 \right\rangle}   
\newcommand{\bra}[1]{\left\langle #1 \right|}   

\newcommand{\lbphi}{\Phi}
\newcommand{\Nmu}{N\mspace{-1mu}(\mspace{-1mu}\mu\mspace{-1mu})}

\newcommand{\sgn}{\operatorname{sgn}}

\newcommand{\rst}{\! \upharpoonright \!} 

\newtheorem{thm}{Theorem}[section]

\newtheorem{cor}[thm]{Corollary}
\newtheorem{lemma}[thm]{Lemma}

\theoremstyle{remark}
\newtheorem*{remark}{Remark}
\newtheorem*{remarks}{Remarks}

\font\notefont=cmsl8 
\title{High Density Limit of the Fermi Polaron with Infinite Mass}
\author{Ulrich Linden\footnote{ulrich.linden@mathematik.uni-stuttgart.de} \, and David 
Mitrouskas\footnote{david.mitrouskas@mathematik.uni-stuttgart.de}\\  
\small Fachbereich Mathematik, Universit\"at Stuttgart, D-70569 Stuttgart, Germany}  

\begin{document}
\maketitle

\frenchspacing

\begin{abstract}
We analyze the ground state energy for $N$ identical fermions in a two-dimensional box of volume $L^2$ interacting with an external point scatterer. Since the 
point scatterer can be considered as an impurity particle of infinite mass, this system is a limit case of the Fermi polaron. We prove that its ground state 
energy in the limit of high density $N/L^2\gg 1$ is given by the polaron energy. The polaron energy is an energy estimate based on trial states up to 
first order in particle-hole expansion, which was proposed by F. Chevy \cite{Chevy} in the physics literature. The relative error in our result is shown to 
be 
small uniformly in $L$. Hence, we do not require a gap of fixed size in the spectrum of the Laplacian on the box. The strategy of our proof relies on a twofold 
Birman-Schwinger type argument applied to the many-particle Hamiltonian of the system.
\end{abstract}

\section{Introduction}
\label{Sect:Intro}

We consider a gas of $N$ identical fermions in a two-dimensional box $\Omega = [-\tfrac{L}{2}, \tfrac{L}{2}]^2$ with 
periodic boundary conditions. The fermions 
do not interact with each other, but with an attractive point scatterer centered at the origin. The Hamiltonian of the 
system acting on the Hilbert space 
$\HH_N := \bigwedge^N L^2(\Omega)$ of anti-symmetric $N$-particle functions is formally given by
\begin{equation} \label{Formal_Expression_Hamiltonian}
   \sum_{i=1}^N \left( -\Delta_{x_i} - g \delta(x_i) \right),
\end{equation}
where $\delta(x_i)$ denotes a Dirac-$\delta$-potential and $g$ plays the role of a coupling constant. This model can be 
seen as the \textit{Fermi polaron} 
in the limit case with an infinitely heavy impurity. The model known as Fermi polaron in the physics literature 
describes a system of $N$ identical fermions 
interacting with a distinct particle by two-body point interactions. It describes an impurity immersed in a gas of 
ultracold fermions.

We are interested in the asymptotics of the ground state energy of the system formally given by 
\eqref{Formal_Expression_Hamiltonian} in the case of high 
density $\rho = N/L^2 \gg 1$. For that we fix a parameter $\mu>0$ and choose the number of fermions by
\begin{equation} \label{Condition_N}
   N = \Nmu := \left| \{ k \in \frac{2 \pi}{L} \Z^2 \: : \: k^2 \leq \mu \} \right|.
\end{equation}
Since \eqref{Condition_N} ensures that the number of fermions coincides with the number of 
eigenvalues of $-\Delta$ that are less or equal than $\mu$, counting multiplicities, the parameter $\mu$ plays the role 
of the \textit{Fermi energy}. The high density regime is then equivalent to $\mu \gg 1$, since $$\rho(\mu) := N(\mu)/L^2 
= \mu/(4\pi) + \mathcal{O}(\sqrt{\mu})$$ as $\mu \to \infty$.\footnote{This formula holds for every fixed $L > 0$, cf. 
\eqref{Estimate_Eta_mu}.} 

We define the self-adjoint Hamiltonian corresponding to \eqref{Formal_Expression_Hamiltonian} by
\begin{align}\label{Second_Quantization}
H_\mu : = d\Gamma(h) \rst \HH_{\Nmu},
\end{align} 
the second quantization (restricted to $\HH_{\Nmu}$) of a one-particle operator $h$ on $L^2(\Omega)$ formally 
represented by $-\Delta_x - g \delta(x)$. The latter means that $h$ restricts to $-\Delta$ on 
$C_0^\infty(\Omega\setminus\{0\})$. All self-adjoint extensions of $-\Delta \rst C_0^\infty(\Omega\setminus\{0\})$ 
are well-known to form a one-parameter family of self-adjoint operators, whose resolvents are rank-one perturbations of 
the resolvent of the Laplacian. All operators $h$ out of this family, which do not coincide with the Laplacian, can be 
parametrized by their ground state energy $E_B < 0$ such that the resolvent reads
\begin{equation} \label{One_particle_Hamiltonian}
   (h - z)^{-1} = (-\Delta - z)^{-1} + \frac{1}{(E_B - z) \sprod{\gamma_{E_B}}{\gamma_{z}}} \ket{\gamma_z} 
\bra{\gamma_{\overline{z}}},
\end{equation}
where $\gamma_z \in L^2(\Omega)$ has Fourier coefficients $\hat{\gamma}_z(k) = (k^2 - z)^{-1}$ for $k \in \frac{2\pi}{L} 
\Z^2$. The spectrum 
of $h$ is purely discrete and has precisely one negative eigenvalue $E_B$. We consider $E_B$ as the model parameter 
which characterizes the interaction strength of the point scatterer.\footnote{For a comprehensive discussion of Laplace 
operators 
with point scatterers, we refer the reader to \cite{Solvable_Models}.}

Our goal is to study the ground state energy $E(\mu) := \min \sigma (H_\mu)$ for $\mu\gg 1$. We prove a conjecture from 
the physics literature stating that $E(\mu)$ is asymptotically given by  $E_0(\mu) + e_P(\mu)$, where
$$ 
   E_0(\mu) = \sum_{k^2 \leq \mu} k^2
$$
is the lowest eigenvalue of $d\Gamma(-\Delta) \rst \HH_{\Nmu}$, and the \textit{polaron energy} $e_P(\mu)$ is the lowest 
solution to the \textit{polaron equation}
\begin{equation} \label{Polaron_Equation}
   - e_P(\mu) = \frac{1}{L^2} \sum_{k^2 \leq \mu} G_\mu(- k^2 - e_P(\mu))^{-1}.
\end{equation}
Here
$$
   G_\mu(\tau) := \frac{1}{L^2} \sum_{k} \left( \frac{1}{k^2 - E_B} - \frac{\chi_{(\mu,\infty)}(k^2)}{k^2 + \tau} 
\right),
$$
for $\tau > -\mu$, where $\chi$ denotes the characteristic function, is a monotonically 
increasing function of $\tau$. We use the convention that all sums and products with respect to $k$ (or other momentum 
variables) run over the momentum space lattice $\frac{2\pi}{L} \Z^2$ unless specified otherwise. By the 
arguments in \cite[Section 8]{GL_Variational}, it is not difficult to see that \eqref{Polaron_Equation} has in fact a 
lowest solution $e_P(\mu)$, which satisfies 
$e_P(\mu) < 0$. 

Our \textbf{main result} states that 
\begin{equation} \label{Formal_description_result}
   E(\mu) = E_0(\mu) +  e_P(\mu) + \mathcal O\left( \frac{|e_P(\mu)|}{\sqrt{\log \mu}}\right)
\end{equation}
as $\mu \to \infty$. Moreover, we show that for large $\mu$, the asymptotics of $e_P(\mu)$ is given by
\begin{align} \label{Formal_estimate_Polaron_energy}
   e_P(\mu) = -\frac{\mu}{\log \mu} + \mathcal{O} \left(\frac{\mu \cdot \log \log \mu}{(\log \mu)^2}\right).
\end{align}

The polaron equation \eqref{Polaron_Equation} was first obtained by F. Chevy using a formal variational calculation 
\cite{Chevy}, in which the trial states were chosen to be one particle-hole excitations of the \textit{Fermi sea}
$$
   \ket{\text{FS}_\mu} = \prod_{k^2 \leq \mu} a_k^* \ket{\text{vac}}.
$$
A rigorous proof of the upper bound $E(\mu) \le E_0(\mu) + e_P(\mu)$ was recently presented in \cite{GL_Variational} 
using a generalized Birman-Schwinger principle for the many-body operator. While the variational principle allows for a 
derivation of an upper bound, it is more involved to obtain a suitable lower bound. A natural approach to derive a 
lower bound for $E(\mu)$ is to compare the eigenvalues $E_B\le \lambda_2 \le \lambda_3 \le...$ and $\lambda_1^0\le 
\lambda_2^0 \le \lambda_3^0 \le...$ of $h$ and $-\Delta$, respectively. In Appendix C we show that 
$(h-z)^{-1}-(-\Delta-z)^{-1}$ being of rank one implies
\begin{equation} \label{Eigenvalue_estimates_rankone_perturbation}
   \lambda^0_i \leq \lambda_{i+1}.
\end{equation}
Hence by rewriting 
$$
   E(\mu) = E_B + \sum_{i=2}^{N(\mu)}\lambda_i \hspace{1cm} \text{and} \hspace{1cm} E_0(\mu) = \sum_{i=1}^{N(\mu)} 
\lambda_i^0,
$$ 
we see that $E(\mu)$ compared to $E_0(\mu)$ is at most lowered by $E_B -\lambda^0_{N(\mu)} \ge E_B-\mu$, i.e.
\begin{align} \label{Naive_Lower_Bound}
E(\mu) \ge E_0(\mu) + E_B - \mu.
\end{align}
As anticipated by \eqref{Formal_description_result} and \eqref{Formal_estimate_Polaron_energy}, this does not provide 
the correct $\mu$-dependence for $\mu\gg 1$. In order to derive an improved lower bound we start with the many-body 
Birman-Schwinger operator for $H_\mu$ discussed in \cite{GL_Variational}. To this operator, we apply a second 
Birman-Schwinger type argument. Together with suitable estimates this is used to show that $E(\mu) - E_0(\mu)$ is 
bounded from below by the solution of a perturbed polaron equation, cf.\ \eqref{Perturbed_Polaron_Equation}. The 
analysis of the large $\mu$ asymptotics of this lower bound leads to the improved estimate. Note that the upper bound 
$E(\mu) \leq E_0(\mu) + e_P(\mu)$ as well as \eqref{Naive_Lower_Bound} can be equally derived for the analogous model in 
three dimensions. Although we expect an asymptotic result similar to \eqref{Formal_description_result} to hold in this 
case, it is not a direct consequence of the method presented in this article.

In the physics literature, the polaron energy is considered a good approximation to the ground state energy of the Fermi 
polaron in the high density regime 
$\mu \gg 1$ as well as in the weak coupling limit $|E_B| \ll 1$. In the strong coupling regime $|E_B| \gg 1$, it is 
expected that one fermion is tightly bound 
to the impurity particle. This behavior is represented by the so-called molecule ansatz \cite{ChevyMora,PDZ}. The two 
classes of trial states were 
investigated by numerical and analytical methods leading to indications for the anticipated difference between the shape 
of the ground state in the weak and 
strong coupling case \cite{GL_Variational, Linden, PS, Parish, Parish_Levinsen, Bruun_Massignan, Schmidt_Enss_1, 
Schmidt_Enss_2, Combescot_Giraud}. For this reason, the Fermi polaron is discussed in the context of the so-called 
BCS-BEC crossover.

We remark that the construction of a semi-bounded Hamiltonian for the Fermi polaron with an impurity of finite mass is 
much more involved compared to \eqref{Second_Quantization}, since it is not a simple generalization of a one-body 
operator. The problem was solved in two \cite{DFT, Dimock_Rajeev, GL_Variational, GL_Stability} and partially in three 
space dimensions \cite{Moser_Seiringer, CDFMT, Minlos}. Rigorous results concerning the ground state energy of these 
models mostly adressed the question of stability and the existence of a lower bound to the Hamiltonian that is uniform 
in the particle number $N$. A recent result shows that the energy shift of the three-dimensional Fermi polaron with an 
impurity of finite mass compared to the non-interacting system can be bounded by an expression depending only on the 
average density and the interaction strength \cite{Moser_Seiringer_2}. The question whether the polaron energy provides 
the correct asymptotics of the ground state energy similar to \eqref{Formal_description_result} remains an open problem.

This article is organized as follows. In Section \ref{Sect:Main_Result} we state and discuss our main result about 
$E(\mu)$ in the high density limit. In Section \ref{Sect:Definition} we state a Birman-Schwinger type principle for 
the many-body Hamiltonian $H_\mu$, which serves as the starting point of our analysis. A proof of this principle based 
on recent results from \cite{GL_Variational} is postponed to Appendix A. Section \ref{Sect:Definition} also includes 
the proof of $E(\mu) \leq E_0(\mu) + e_P(\mu)$ which is obtained from a variational calculation for the generalized 
Birman-Schwinger operator. A suitable lower bound for $E(\mu)$ is established in Section \ref{Sect:Birman_Schwinger}, 
and in Section \ref{Sect:High_Density_Limit} we analyze its asymptotics in the high density regime to conclude the 
proof of \eqref{Formal_description_result}.

\section{Main result}
\label{Sect:Main_Result}

\begin{thm} \label{Main_Theorem}
For given $L > 0$, $\mu > 0$ and $E_B < 0$ let the number $\Nmu$ of particles be fixed by \eqref{Condition_N}, and let the Hamiltonian $H_\mu$ 
be defined by \eqref{Second_Quantization} and \eqref{One_particle_Hamiltonian}. Then, the ground state energy $E(\mu) = 
\min\sigma(H_\mu)$ and the lowest 
solution of the 
polaron equation \eqref{Polaron_Equation}, $e_P(\mu) < 0$, satisfy the following property. There are constants $c, C > 0$ such that
\begin{equation} \label{Asymptotics}
   \left| E(\mu) - E_0(\mu) - e_P(\mu) \right| \leq C \cdot \frac{|e_P(\mu)|}{\sqrt{\log(\mu/|E_B|)}}
\end{equation}
for all $\mu/|E_B| \geq c$ and $L^2 \cdot |E_B| \geq 1$.
\end{thm}

\begin{remarks}\leavevmode
\begin{enumerate}[label=(\roman*)]
\item (Polaron energy) In Lemma \ref{Lemma_bound_for_e_p} we show that the polaron energy $e_P(\mu)$ satisfies
\begin{equation} \label{Polaron_Energy_Asymptotics}
  e_P(\mu) = - \frac{\mu}{\log(\mu/|E_B|)} \cdot \left(1 + \mathcal{O}\left(\frac{\log\log(\mu/|E_B|)}{\log(\mu/|E_B|)} \right) \right)
\end{equation}
as $\mu/|E_B| \to \infty$. The remainder term is uniformly bounded for $L^2 \cdot |E_B| \geq 1$.

\item ($L^2\cdot |E_B| \ge 1 $) We suggest to read this condition as a characterization of the interacting regime of the considered model. For a 
two-dimensional box with periodic boundary conditions and side length $L$, the smallest excitation energy of the kinetic energy operator equals $4\pi^2 / L^2$. 
This quantity is at most of the same order of magnitude as the two-body binding energy $|E_B|$ as long as $L^2 \cdot |E_B|\ge 1$.

\item (Thermodynamic limit) We stress that our result is not based on the assumption of a fixed spectral gap of the Laplacian, which can be read off the fact 
that the error terms in \eqref{Asymptotics} and \eqref{Polaron_Energy_Asymptotics} are uniformly bounded for $L$ large enough. This allows us to make the 
following assertion about the ground state energy in the thermodynamic limit. Assuming the existence of the limits of $E(\mu)-E_0(\mu)$ and $e_P(\mu)$ as $L 
\to \infty$ with $\mu > 0$ kept fixed, our result implies
\begin{equation} \label{TD_Limit}
   \lim_{L \to \infty} \left( E(\mu) - E_0(\mu) \right) = e_P^{\text{\tiny{TD}}}(\mu) + \mathcal{O} \left( 
\frac{|e_P^{\text{\tiny{TD}}}(\mu)|}{\sqrt{\log(\mu/|E_B|)}} \right)
\end{equation}
as $\mu/|E_B| \to \infty$ with $e_P^{\text{\tiny{TD}}}(\mu) := \lim_{L\to\infty} e_P (\mu)$. Even though the existence of the limits is naturally 
expected, we do not pursue to prove this in the present work.

\item[(iv)] (Weak coupling limit) By the discussion of the thermodynamic limit in the previous remark, the polaron energy also approximates 
the ground state energy in the thermodynamic limit as $E_B \nearrow 0$, cf. \eqref{TD_Limit}. This confirms another 
conjecture from the physics literature (see 
e.g. \cite{Parish}).
\end{enumerate}
\end{remarks}

\section{Preliminaries and Upper Bound}
\label{Sect:Definition}

In order to estimate the ground state energy $E(\mu)$, we consider the many-body operator $H_\mu$ directly instead of 
evaluating the eigenvalue sum of the 
one-body operator $h$. The starting point for our proof of Theorem \ref{Main_Theorem} is a criterion for upper and lower bounds of $E(\mu)$ in terms of a 
many-body Birman-Schwinger operator which we denote by $\phi_\mu(\lambda)$. It states that for $\lambda < E_0(\mu)$
\begin{equation} \label{Upper_Lower_Bound_Criterion}
   E(\mu) \leq \lambda \qquad \Leftrightarrow \qquad \inf\sigma(\phi_\mu(\lambda)) \leq 0,
\end{equation}
where equality on one side implies equality on both sides. A proof of this equivalence based on the results in \cite{GL_Variational} (in 
particular Theorem 6.3 and Corollary 6.4) is included in Appendix A. Here we give an explicit expression for 
$\phi_\mu(\lambda)$, which will be the main object to be analyzed.

In this regard it is convenient to use the formalism of second quantization. Let $a_k$ and $a_k^*$ for $k \in 
\frac{2\pi}{L} \Z^2$ be the usual fermionic creation and annihilation operators of plane waves $\varphi_k$ 
with
$$
   \varphi_k(x) := L^{-1} \cdot \exp(ikx).
$$
The operator $T := \sum_k \: k^2 \: a_k^* a_k$ on the antisymmetric Fock space over $L^2(\Omega)$ representing the 
kinetic energy is the second quantization of the 
Laplacian. The self-adjoint operators $\phi_\mu(\lambda)$ for $\lambda < E_0(\mu)$ form an analytic family of type (A), and for $\lambda < 0$ they are given by 
\begin{equation} \label{Def:Phi}
  \phi_\mu(\lambda) = \frac{1}{L^2} \sum_k \left( \frac{1}{k^2 - E_B} - \frac{1}{T + k^2 - \lambda} \right) + \frac{1}{L^2} \sum_{k,l} a_k^* \, \frac{1}{T + 
k^2 + l^2 - \lambda} \, a_l
\end{equation}
on $\HH_{\Nmu-1}$. The first operator in \eqref{Def:Phi} is an unbounded function of $T$, whereas the second term is a bounded 
operator on $\HH_{\Nmu-1}$.

For the proof of Theorem \ref{Main_Theorem} we derive an upper and a lower bound for $E(\mu)$ with the help of \eqref{Upper_Lower_Bound_Criterion}.\medskip

\noindent The \textbf{upper bound},
\begin{equation} \label{Upper_Bound}
   E(\mu) \leq E_0(\mu) + e_P(\mu),
\end{equation}
follows by the same arguments as in \cite[Section 8]{GL_Variational}. One considers the trial state
$$
   w := \sum_{q^2 \leq \mu} G_\mu(- q^2 - e_P(\mu))^{-1} a_q \ket{\text{FS}_\mu} \in D(\phi_\mu),
$$
and uses \eqref{Phi_Rewritten} and \eqref{Polaron_Equation} to verify
$$
   \sprod{w}{\phi_\mu(E_0(\mu) + e_P(\mu)) \, w} = 0.
$$
By \eqref{Upper_Lower_Bound_Criterion}, this implies $E(\mu) \leq E_0(\mu) + e_P(\mu)$.\medskip

\noindent The rest of this article is devoted to the proof of a corresponding \textbf{lower bound} for $E(\mu)$.

\section{Twofold Birman-Schwinger argument}
\label{Sect:Birman_Schwinger}

In order to estabilsh a lower bound to $E(\mu)$ with the help of \eqref{Upper_Lower_Bound_Criterion}, we derive a lower 
bound for the many-body Birman-Schwinger operator $\phi_\mu(\lambda)$. In view of \eqref{Upper_Bound}, it suffices to 
consider $\lambda \leq E_0(\mu) + e_P(\mu)$.

\begin{lemma} \label{Lower_Bound_Phi}
Let $\lambda \leq E_0(\mu) + e_P(\mu)$. Then, $\phi_\mu(\lambda) \geq \lbphi_\mu(\lambda)$, where 
$$
  \lbphi_\mu(\lambda) := \left( G_\mu(T - \lambda) - r_{\!\mu,\lambda} - a(\eta_\mu) \frac{1}{T - \lambda} a^*(\eta_\mu) 
\right) \rst \HH_{\Nmu-1}.
$$
and $\eta_\mu := L^{-1} \cdot \sum_{k^2 \leq \mu} \ph_k$. The constant $r_{\!\mu,\lambda}$ is defined by
\begin{equation} \label{definition_konstante}
 r_{\!\mu,\lambda} := 2 \cdot \norm{a(\eta_\mu) A_{\mu,\lambda}^*},
\end{equation}
where the operator $A_{\mu,\lambda} \in \LL(\HH_{\Nmu},\HH_{\Nmu-1})$ for $\lambda < E_0(\mu)$ is given by
$$
   A_{\mu,\lambda} := \frac{1}{L} \cdot \sum_{k^2 > \mu} \frac{1}{T + k^2 - \lambda} \, a_k \rst \HH_{\Nmu}.
$$
\end{lemma}

\begin{remark}
The convergence of the series defining $A_{\mu,\lambda}$ and the boundedness of this operator can be verified by the 
Cauchy-Schwarz inequality and $\sum_{k} a_k^* a_k \rst \HH_N = N$, since for $\psi \in \HH_N$
\begin{align*}
   \sum_{k^2 > \mu} \norm{(T + k^2 - \lambda)^{-1} \, a_k \psi} 
   &\leq \sum_{k^2 > \mu} \frac{\norm{a_k \psi}}{E_0(\mu) - \mu + k^2 - \lambda} \\
   &\leq \sqrt{N} \cdot \norm{\psi} \cdot \left(\sum_{k^2 > \mu} \frac{1}{(E_0(\mu) - \mu + k^2 - \lambda)^2} \right)^{1/2}
   < \infty.
\end{align*}
\end{remark}

\begin{proof}
We rewrite \eqref{Def:Phi} by reversing the normal ordering of the $a_k$ and $a_k^*$ for all $k \in \frac{2\pi}{L} \Z^2$ with $k^2 \leq \mu$. Using the 
pull-through formula, which reads
\begin{equation} \label{Pull_Through_Formula}
   a_k f(T) \psi = f(T + k^2) a_k \psi
\end{equation}
for $\psi \in \HH_N$ and bounded continuous functions $f: [\min \sigma(T \rst \HH_N), \infty) \to \C$, we obtain for $\lambda < 0$
\begin{align}
   \phi_\mu(\lambda) &= G_\mu(T - \lambda) - a(\eta_\mu) \frac{1}{T - \lambda} a^*(\eta_\mu) \nonumber\\ 
   &\qquad - a(\eta_\mu) \, A_{\mu,\lambda}^* - A_{\mu,\lambda} \, a^*(\eta_\mu) + \frac{1}{L^2} \sum_{k^2,l^2 > \mu} 
a_k^* \, \frac{1}{T + k^2 + l^2 - \lambda} \, a_l. \label{Phi_Rewritten}
\end{align}
on $\HH_{\Nmu-1}$. Note that \eqref{Phi_Rewritten} applied to a vector $\psi \in D(\phi_\mu) \subseteq \HH_{\Nmu-1}$ is an analytic function of $\lambda$ for 
$\lambda < E_0(\mu)$. Since $\phi_\mu(\lambda)$ is an analytic family of type (A), \eqref{Phi_Rewritten} also holds for $\lambda \leq E_0(\mu) + e_P(\mu)$.

The last term in \eqref{Phi_Rewritten} is positive for $\lambda \leq E_0(\mu) + e_P(\mu)$, since on $\HH_{\Nmu-1}$
\begin{align}
 \sum_{k^2,l^2 > \mu} a_k^* \, \frac{1}{T + k^2 + l^2 - \lambda} \, a_l
  &= \int\limits_0^\infty dt \: \sum_{k^2,l^2 > \mu} a_k^* \: e^{-t(T + k^2 + l^2 - \lambda)} \: a_l \nonumber\\
  &= \int\limits_0^\infty dt \: \left( \sum_{k^2 > \mu} e^{-t k^2} a_k^* \right) e^{-t(T - \lambda)} \left( \sum_{l^2 > \mu} e^{-t l^2} a_l \right) \geq 0. 
\label{Easy_Positivity_Argument}
\end{align}
This proves the statement of the lemma.
\end{proof}

To derive a lower bound for $\lbphi_\mu(\lambda)$, we make use of the fact that this operator is given by an expression of the form $K - V^* V$ with $K = 
G_\mu(T - \lambda) - r_{\!\mu,\lambda}$ and a bounded operator $V = (T - z)^{-1/2} a^*(\eta_\mu)$. If $K$ is self-adjoint and $K \geq c > 0$, it 
follows readily
\begin{align}
   K - V^* V &= (K - V^* V) K^{-1} (K - V^* V) + V^* (1- V K^{-1} V^*) V \nonumber \\ &\geq V^* (1- V K^{-1} V^*) V. \label{Formal_Birman_Schwinger_Argument}
\end{align}
This is a key argument in the proof of the following lemma.

\begin{lemma} \label{Second_Birman_Schwinger}
 Let $\lambda \leq E_0(\mu) + e_P(\mu)$ and suppose that $G_\mu(-\mu-e_P(\mu)) > r_{\!\mu,\lambda}$. Then, $\lbphi_\mu(\lambda) \geq 0$ if
\begin{equation} \label{Perturbed_Polaron_Equation}
   E_0(\mu) - \lambda - \frac{1}{L^2} \sum_{k^2 \leq \mu} (G_\mu(E_0(\mu) - k^2 - \lambda) - r_{\!\mu,\lambda})^{-1} = 0.
\end{equation}
\end{lemma}

\begin{remark}
 Neglecting $r_{\!\mu,\lambda}$, \eqref{Perturbed_Polaron_Equation} reduces to the polaron equation \eqref{Polaron_Equation}. For this reason, we refer to 
\eqref{Perturbed_Polaron_Equation} as \textit{perturbed polaron equation}.
\end{remark}

\begin{proof}
Since $T \geq E_0(\mu) - \mu$ on $\HH_{\Nmu-1}$ and $\lambda \leq E_0(\mu) + e_P(\mu)$, it follows that $G_\mu(T - \lambda) \geq G_\mu(- \mu - e_P(\mu))$ on 
$\HH_{\Nmu-1}$, which is strictly larger than $r_{\!\mu,\lambda}$ by assumption. Thus by \eqref{Formal_Birman_Schwinger_Argument},
\begin{equation} \label{Phi_F_Bound}
   \lbphi_\mu(\lambda) \geq a(\eta_\mu) \: \frac{1}{T-z} \: \mathcal{F}_\mu(\lambda) \: \frac{1}{T-z} \: a^*(\eta_\mu) 
\: \rst \HH_{\Nmu - 1},
\end{equation}
where
$$
   \mathcal{F}_\mu(\lambda) = \left( T - \lambda - a^*(\eta_\mu) \: (G_\mu(T - \lambda) - r_{\!\mu,\lambda})^{-1} \:
a(\eta_\mu) \right) \rst 
\HH_{\Nmu}
$$
is a self-adjoint operator with $D(\mathcal{F}_\mu(\lambda)) = D(T \rst \HH_{\Nmu})$. For the derivation of a lower bound we 
approximate $\mathcal{F}_\mu(\lambda)$ by $\mathcal{F}_\mu^{(n)}(\lambda)$. This operator arises from $\mathcal{F}_\mu(\lambda)$ by replacing the 
function $G_\mu$ by $G_\mu^{(n)}$ with
$$
   G_\mu^{(n)}(\tau) = \frac{1}{L^2} \sum_{k^2 \leq n} \left( \frac{1}{k^2 - E_B} - 
\frac{\chi_{(\mu,\infty)}(k^2)}{k^2 + \tau} \right).
$$
Note that $G_\mu^{(n)}(\tau) \to G_\mu(\tau)$ as $n \to \infty$ for every $\tau > -\mu$. Thus, $G_\mu^{(n)}(T-\lambda) \psi \to G_\mu(T-\lambda) \psi$ as $n 
\to \infty$ for every antisymmetric product state $\psi \in \HH_{\Nmu - 1}$ of plane waves. Since these states form a total set of eigenstates of 
$G_\mu(T-\lambda)$ on $\HH_{\Nmu-1}$, the set $D \subseteq \HH_{\Nmu-1}$ of their finite linear combinations is a domain of essential self-adjointness for this 
operator. Furthermore,  $G_\mu^{(n)}(T-\lambda) \geq G_\mu^{(n)}(-\mu - e_P(\mu))$. By the convergence of $G_\mu^{(n)}$ and the assumption $G_\mu(-\mu-e_P(\mu)) 
> r_{\!\mu,\lambda}$, there is $\eps > 0$ such that $G_\mu^{(n)}(T-\lambda) - r_{\!\mu,\lambda} > \eps$ for $n$ large enough. Hence as $n \to \infty$, 
$(G_\mu^{(n)}(T-\lambda) - r_{\!\mu,\lambda})^{-1} \to (G_\mu(T-\lambda) - r_{\!\mu,\lambda})^{-1}$ and $\mathcal{F}_\mu^{(n)}(\lambda) \to 
\mathcal{F}_\mu(\lambda)$ strongly.

Using the pull-through formula \eqref{Pull_Through_Formula}, we rewrite $\mathcal{F}_\mu^{(n)}(\lambda)$ as
\begin{equation}
   T - \lambda - \frac{1}{L^2} \sum_{k^2 \leq \mu} (G_\mu^{(n)}(T - k^2 - \lambda) - r_{\!\mu,\lambda})^{-1}
   + \frac{1}{L^2} \sum_{k^2,l^2 \leq \mu} \!\! a_k (G_\mu^{(n)}(T - k^2 - l^2 - \lambda) - r_{\!\mu,\lambda})^{-1} a_l^* \label{T_lambda_n}
\end{equation}
on $\HH_{\Nmu}$. Assuming that the last term in \eqref{T_lambda_n}, which we call $\mathcal{P}_\mu^{(n)}(\lambda)$ in the following, is a positive operator on 
$\HH_{\Nmu}$, we obtain 
$$
   \mathcal{F}_\mu^{(n)}(\lambda) \geq E_0(\mu) - \lambda - \frac{1}{L^2} \sum_{k^2 \leq \mu} (G_\mu^{(n)}(E_0(\mu) - k^2 - \lambda) - r_{\!\mu,\lambda})^{-1},
$$
since $T \geq E_0(\mu)$ on $\HH_{\Nmu}$. In view of \eqref{Phi_F_Bound}, this lower bound completes the proof, since it converges to the left hand side of 
\eqref{Perturbed_Polaron_Equation} as $n \to \infty$.

It remains to prove $\mathcal{P}_\mu^{(n)}(\lambda) \geq 0$. For $\psi \in \HH_{\Nmu}$,
\begin{align*}
   L^2 \cdot \sprod{\psi}{\mathcal{P}_\mu^{(n)}(\lambda) \psi} 
   &= \int\limits_0^\infty dt \sum_{k^2,l^2 \leq \mu} \sprod{\psi}{a_k \: \exp(-t \cdot [ G_\mu^{(n)}(T - k^2 - l^2 - 
\lambda) - r_{\!\mu,\lambda}]) \: a_l^* \psi} \\
   &= \int\limits_0^\infty dt \: \exp(-t [L^{-2} \sum_{p^2 \leq n} \tfrac{1}{p^2 - E_B} \!-\! r_{\!\mu,\lambda}]) \cdot \mathcal{I}^{(n)}(t)
\end{align*}
with
$$
   \mathcal{I}^{(n)}(t) := \sum_{k^2,l^2 \leq \mu} \sprod{\psi}{a_k 
\prod\limits_{\mu < q^2 \leq n} \!\! \exp(t (q^2 + T - k^2 - l^2 - \lambda)^{-1}) \:\, a_l^* \psi}.
$$
We show that $\mathcal{I}^{(n)}(t) \geq 0$ for all $t \in [0,\infty)$ and $n \in \N$. Note that the product in the definition of $\mathcal{I}^{(n)}(t)$ has 
only finitely many factors, because $A_n := \{ q \in \frac{2\pi}{L} \Z^2 \: | \: \mu < q^2 \leq n \}$ is a finite set. 
We consider the exponential series and obtain
$$
   \mathcal{I}^{(n)}(t) = \sum_{k^2,l^2 \leq \mu} \sprod{\psi}{a_k \prod\limits_{q \in A_n} \!\! \left( \sum\limits_{m=0}^\infty \frac{t^m}{m!} \cdot 
\frac{1}{(q^2 + T - k^2 - l^2 - \lambda)^m} \right) a_l^* \psi}.
$$
By the absolute convergence of the exponential series, we can rearrange the product of series to get
$$
   \mathcal{I}^{(n)}(t) = \sum_{m: A_n \to \N_0} \:\: \sum_{k^2,l^2 \leq \mu} \sprod{\psi}{a_k \prod\limits_{q \in A_n} \!\! \left( \frac{t^{m(q)}}{m(q)!} 
\cdot 
\frac{1}{(q^2 + T - k^2 - l^2 - \lambda)^{m(q)}} \right) a_l^* \psi},
$$
where we sum over all $\N_0$-valued functions $m$ on the finite set $A_n$. Note that the factor in parantheses indexed 
by $q$ is equal to $1$ if $m(q) = 0$. For all factors with $m(q) \neq 0$, we use the identity
$$
   \frac{1}{a^\tau} = \frac{1}{c_\tau} \cdot \int\limits_0^\infty \! ds \: e^{-a s^{1/\tau}} \qquad \text{with} \qquad 
c_\tau := 
\int\limits_0^\infty \! ds \: e^{- s^{1/\tau}}
$$
for $a,\tau > 0$ to rewrite each of the summands in the $m$-sum as
\begin{align*}
 &\sum_{k^2,l^2 \leq \mu} \:\: \prod\limits_{\substack{q \in A_n \\ m(q) \neq 0}} \: \left( \frac{c_{m(q)} t^{m(q)}}{m(q)!} 
\int\limits_0^\infty \!ds_q \right) \sprod{\psi}{a_k \prod\limits_{\substack{p \in A_n \\ m(p) \neq 0}} \!\! e^{-(p^2 + T - k^2 - l^2 - \lambda)s_p^{1/m(p)}} 
a_l^* \psi} \\
 &=\prod\limits_{{\substack{q \in A_n \\ m(q) \neq 0}}} \left( \frac{c_{m(q)} t^{m(q)}}{m(q)!} \int\limits_0^\infty \!ds_q \right) \norm{\sum\limits_{k^2 \leq 
\mu} \prod\limits_{{\substack{p \in A_n \\ m(p) \neq 0}}} e^{-\frac{1}{2}(p^2 + T - 2k^2 - \lambda)s_p^{1/m(p)}} a_k^* \psi }^2.
\end{align*}
This yields $\mathcal{I}^{(n)}(t) \geq 0$. Thus, $\sprod{\psi}{\mathcal{P}_\mu^{(n)}(\lambda) \psi} \geq 0$ for all $n \in \N$.
\end{proof}

To summarize Section \ref{Sect:Birman_Schwinger}, we combine \eqref{Upper_Lower_Bound_Criterion}, Lemma \ref{Lower_Bound_Phi} and Lemma 
\ref{Second_Birman_Schwinger} to obtain the following statement. Recall that equality on one side of \eqref{Upper_Lower_Bound_Criterion} implies equality on 
both sides.

\begin{cor} \label{Summarizing_Lemma}
Suppose that $\lambda(\mu) \leq E_0(\mu) + e_P(\mu)$ satisfies the perturbed polaron equation 
\eqref{Perturbed_Polaron_Equation} and assume $G_\mu(-\mu-e_P(\mu)) > r_{\!\mu,\lambda}$. Then, $H_\mu \geq 
\lambda(\mu)$.
\end{cor}

\section{High density limit}
\label{Sect:High_Density_Limit}

The main result of this section is the following lemma, which, together with Corollary \ref{Summarizing_Lemma}, concludes the proof of 
Theorem \ref{Main_Theorem}. For notational convenience we set $\widetilde{\mu} := \mu/|E_B|$ and $\widetilde{L} := L \cdot \sqrt{|E_B|}$ throughout this 
section.

\begin{lemma}\label{Lemma_perturberd_polaron_bound}\leavevmode
\begin{enumerate}[label=(\alph*)]
 \item If $\widetilde{\mu}$ is large enough, then $G_\mu(-\mu-e_P(\mu)) > r_{\!\mu,\lambda}$ for all $\lambda \leq E_0(\mu) + e_P(\mu)$ and $\widetilde{L} \geq 1$.
 \item If $\widetilde{\mu}$ is large enough and $\widetilde{L} \geq 1$, there exists a unique solution $\lambda(\mu) \leq E_0(\mu) + e_P(\mu)$ to the perturbed
polaron equation \eqref{Perturbed_Polaron_Equation}. Moreover, there is a constant $C > 0$ satisfying the following property. If $\widetilde{\mu}$ is large 
enough, then
 \begin{align*}
   E_0(\mu) + e_P(\mu) -\lambda(\mu) \leq C \cdot \frac{|e_P(\mu)|}{\sqrt{\log \widetilde{\mu}}}
\end{align*}
 for all $\widetilde{L} \geq 1$.
\end{enumerate}
\end{lemma}

\noindent Before proving Lemma \ref{Lemma_perturberd_polaron_bound}, we state and prove two preparatory lemmas.

\begin{lemma} \label{Lemma_Bound_G_mu}
   Let $\tau > - \mu$ and set $\widetilde{\tau} := \tau/|E_B|$. Then for all $L > 0$,
\begin{equation} \label{Bound_G_mu}
  \left| G_\mu(\tau) - \frac{1}{4 \pi} \log \left( \widetilde{\mu} + \widetilde{\tau} \right) \right|
   \leq \frac{K(\widetilde{\tau}, \widetilde{\mu}, \widetilde{L})}{\widetilde{L}},
\end{equation}
where the function on the right hand side is given by
\begin{align}
   K(\widetilde{\tau}, \widetilde{\mu}, \widetilde{L}) :=  1 + \frac{3}{\widetilde{L}} + \frac{1}{\sqrt{\widetilde{\mu} + 
\min\{\widetilde{\tau},1\}}} + \left( \frac{4 \sqrt{\widetilde{\mu}}}{\pi} + \frac{6}{\widetilde{L}} \right) \frac{1}{\widetilde{\mu} + 
\min\{\widetilde{\tau},1\}}. \label{Definition_K_mu}
\end{align}
\end{lemma}

\begin{proof}
For $ \tau >-\mu$, we write
\begin{equation*}
   G_\mu(\tau) = \frac{1}{L^2} \sum_{k} f_\mu(k^2)  + \sgn(\tau +E_B) \frac{1}{L^2} \sum_{k^2> \mu}  g_\mu(k^2),
\end{equation*}
where $ \sgn(\tau +E_B)$ denotes the sign of $\tau+E_B$, and
\begin{align*}
   f_\mu(k^2):= \frac{\chi_{[0,\mu]}(k^2)}{k^2 - E_B}, \hspace{1cm}
   g_\mu(k^2):= \frac{\vert \tau + E_B\vert }{(k^2 - E_B)(k^2 + \tau)},
\end{align*}
are non-negative, monotonically decreasing functions. Observe that
\begin{align*}
\frac{1}{4\pi} \log 
\left( \frac{\mu + \tau}{|E_B|}  \right) = \frac{1}{2\pi}\int\limits_0^{\infty}  f_\mu(t^2) t \; dt + \frac{\sgn(\tau +E_B)}{2\pi} 
\int\limits_{\sqrt{\mu}}^\infty  g_\mu(t^2)t\; dt.
\end{align*}
Thus, the left hand side of \eqref{Bound_G_mu} can be estimated by the sum of
\begin{align}
\left| \frac{1}{L^2} \sum_{k} f_\mu(k^2)  -  \frac{1}{2\pi}\int\limits_0^{\infty}  f_\mu(t^2) t \; dt \right| \le \frac{2}{\pi L}   \int\limits_0^{\infty}  
f_\mu(t^2) \; dt  + \frac{3 f_\mu(0)}{L^2 }, \label{G_mu_Error_a}
\end{align}
and
\begin{align}
   \left| \frac{1}{L^2} \sum_{k^2>\mu} g_\mu(k^2)  -  \frac{1}{2\pi}\int\limits_{\sqrt{\mu}}^{\infty}  g_\mu(t^2) t \; dt \right|
   \leq \frac{2}{\pi L} \int\limits_{\sqrt{\mu}}^\infty g_\mu(t^2) \; dt + \left(\frac{4\sqrt{\mu}}{\pi L} + \frac{6}{L^2}\right) g_\mu(\mu), 
\label{G_mu_Error_b}
\end{align}
where we use Lemma \ref{Riemann} in \eqref{G_mu_Error_a} and \eqref{G_mu_Error_b}. One concludes the proof by estimating 
the error terms 
in \eqref{G_mu_Error_a} and \eqref{G_mu_Error_b}, which leads to the bound $K(\widetilde{\tau}, \widetilde{\mu}, \widetilde{L})/\widetilde{L}$ in 
\eqref{Bound_G_mu}. For the error term in \eqref{G_mu_Error_a}, we use
\begin{align*}
  \int\limits_0^{\infty} f_\mu(t^2)\; dt \leq \int\limits_0^\infty \: \frac{1}{t^2 - E_B}\; dt = \frac{\pi}{2} \cdot \frac{1}{\sqrt{|E_B|}}.
\end{align*}
For the error term in \eqref{G_mu_Error_b}, consider
$$
    g_\mu(t^2) = \sgn(\tau + E_B) \cdot \left( \frac{1}{t^2 - E_B} - \frac{1}{t^2 + \tau} \right) \leq \frac{1}{t^2 + \min\{\tau,|E_B|\}}.
$$
This implies $g_\mu(\mu) \leq ( \mu + \min\{\tau,|E_B|\} )^{-1}$ as well as
\begin{align*} 
   \int\limits_{\sqrt{\mu}}^\infty g_\mu(t^2) \, dt
   &\leq \int\limits_{\sqrt{\mu}}^\infty \frac{dt}{t^2 + \min\{\tau,|E_B|\}}
   \leq \int\limits_{\sqrt{\mu}}^\infty \frac{dt}{(t - \sqrt{\mu})^2 + \mu + \min\{\tau,|E_B|\}} \\
   &= \frac{\pi}{2} \cdot \frac{1}{\sqrt{\mu + \min\{\tau,|E_B|\}}}
\end{align*}
where we used $t^2 = (t + \sqrt{\mu})(t - \sqrt{\mu}) + \mu \geq (t - \sqrt{\mu})^2 + \mu$ in the denominator.
\end{proof}

\begin{lemma}\label{Lemma_bound_for_e_p}
There is a constant $C > 0$ satisfying the following property. If $\widetilde{\mu}$ is large enough, then
\begin{align}
   \left\vert e_P(\mu) + \frac{\mu}{\log \widetilde{\mu}} \right\vert 
   \leq C \cdot \frac{\mu \cdot \log \log \widetilde{\mu}}{(\log \widetilde{\mu})^2} \label{Bound_for_e_p}
\end{align}
for all $\widetilde{L} \geq 1$.
\end{lemma}

\begin{proof}
We set $z_P := |e_P(\mu)|$ and $\widetilde{z}_P := z_P/|E_B|$ for notational convenience. In the first part of the proof we derive a lower bound for 
$\widetilde{z}_P$. By \eqref{Polaron_Equation} and $G_\mu(- q^2 + z_P) \leq G_\mu(z_P)$, we obtain
\begin{equation} \label{First_Lower_Bound_z_P}
   z_P \geq G_\mu(z_P)^{-1} \Bigg( L^{-2} \sum_{q^2 \leq \mu} 1 \Bigg).
\end{equation}
By \eqref{Naive_Lower_Bound} and \eqref{Upper_Bound}, 
\begin{equation} \label{A_priori_bound_z_P}
   \widetilde{z}_P \leq \widetilde{\mu} + 1.
\end{equation}
Combining this inequality with Lemma \ref{Lemma_Bound_G_mu}, we obtain
$$
   G_\mu(z_P) \leq \frac{1}{4 \pi} \log\widetilde{\mu} + \frac{1}{4 \pi} \log(2 + 1/\widetilde{\mu}) + \frac{K(\widetilde{z}_P, \widetilde{\mu}, 
\widetilde{L})}{\widetilde{L}}.
$$
Since $\widetilde{z}_P \geq 0$, $K(\widetilde{z}_P,\widetilde{\mu},\widetilde{L})$ is uniformly bounded for $\widetilde{\mu}, \widetilde{L} \geq 1$ and it 
follows that there is a constant $C_1 > 0$ such that
$4 \pi G_\mu(z_P) \leq \log(\widetilde{\mu}) + C_1/(4\pi)$ and consequently
\begin{equation} \label{Constant_C_1}
   G_\mu(z_P)^{-1} \geq \frac{4 \pi}{\log\widetilde{\mu}} - \frac{C_1}{(\log\widetilde{\mu})^2}.
\end{equation}
for all $\widetilde{\mu}, \widetilde{L} \geq 1$. By Lemma \ref{Riemann} (a),
\begin{align}
 \Bigg|  L^{-2} \sum_{q^2\le \mu} 1  - \frac{\mu}{4\pi} \Bigg| 
 \leq  \frac{2\sqrt\mu}{\pi L} + \frac{3}{L^2} \label{Estimate_Eta_mu}
\end{align}
By \eqref{First_Lower_Bound_z_P}, \eqref{Constant_C_1} and \eqref{Estimate_Eta_mu}
$$
   \widetilde{z}_P 
   \geq \frac{\widetilde{\mu}}{\log\widetilde{\mu}} - \frac{C_1}{4\pi} \frac{\widetilde{\mu}}{(\log\widetilde{\mu})^2}
   - \frac{4 \pi}{\log\widetilde{\mu}} \left(\frac{2 \sqrt{\widetilde{\mu}}}{\pi\widetilde{L}} + \frac{3}{\widetilde{L}^2} \right)
$$
and hence there is a constant $C_2 > 0$ such that if $\widetilde{\mu}$ is large enough,
\begin{equation} \label{Constant_C_2}
   \widetilde{z}_P 
   \geq \frac{\widetilde{\mu}}{\log\widetilde{\mu}} - C_2 \cdot \frac{\widetilde{\mu}}{(\log\widetilde{\mu})^2}
\end{equation}
for all $\widetilde{L} \geq 1$.

Now, we derive an upper bound for $\widetilde{z}_P$. By \eqref{Polaron_Equation} and $G_\mu(z_P - q^2) \geq G_\mu(z_P - \mu)$ for $q^2 \leq \mu$,
\begin{equation} \label{First_Upper_Bound_z_P}
   z_P \leq G_\mu(z_P - \mu)^{-1} \Bigg( L^{-2} \sum_{q^2 \leq \mu} 1 \Bigg). 
\end{equation}
By Lemma \ref{Lemma_Bound_G_mu},
$$
   G_\mu(z_P - \mu) \geq \frac{1}{4\pi} \log\widetilde{z}_P - \frac{K(\widetilde{z}_P - \widetilde{\mu},\widetilde{\mu},\widetilde{L})}{\widetilde{L}}.
$$
We employ \eqref{Constant_C_2} to obtain the following two statements. Firstly, by \eqref{A_priori_bound_z_P},
$$
   K(\widetilde{z}_P - \widetilde{\mu},\widetilde{\mu},\widetilde{L}) = 1 + \frac{3}{\widetilde{L}} + \frac{1}{\sqrt{\widetilde{z}_P}} + 
\frac{1}{\widetilde{z}_P} \cdot \left( \frac{4 \sqrt{\widetilde{\mu}}}{\pi} + \frac{6}{\widetilde{L}} \right),
$$
which is bounded uniformly for $\widetilde{L} \geq 1$ and $\widetilde{\mu}$ large. Secondly, $\log(\widetilde{z}_P) \geq 
\log(\widetilde{\mu}) - 2 \log \log (\widetilde{\mu})$ for all $\widetilde{L} \geq 1$ if $\widetilde{\mu}$ is large enough. Thus, there is a constant $C_3 > 0$ 
such that if $\widetilde{\mu}$ is large enough
\begin{equation} \label{Constant_C_3_Inverse}
   4 \pi G_\mu(z_P - \mu) \geq \log\widetilde{\mu} - C_3/(2\pi) \cdot \log\log\widetilde{\mu}
\end{equation}
and consequently
\begin{equation} \label{Constant_C_3}
   G_\mu(z_P - \mu)^{-1} \leq \frac{4\pi}{\log\widetilde{\mu}} + C_3 \cdot \frac{\log\log\widetilde{\mu}}{(\log\widetilde{\mu})^2} 
\end{equation}
for all $\widetilde{L} \geq 1$.
By \eqref{First_Upper_Bound_z_P}, \eqref{Constant_C_3} and \eqref{Estimate_Eta_mu} there is a constant $C_4 > 0$ such 
that if $\widetilde{\mu}$ is large enough
\begin{equation*}
   \widetilde{z}_P 
   \leq \frac{\widetilde{\mu}}{\log\widetilde{\mu}} + C_4 \cdot \frac{\widetilde{\mu} \cdot \log\log\widetilde{\mu}}{(\log\widetilde{\mu})^2}
\end{equation*}
for all $\widetilde{L} \geq 1$, and the proof of the lemma is complete
\end{proof}

\begin{proof}[Proof of Lemma \ref{Lemma_perturberd_polaron_bound}] We prove (a) by deriving a bound for $r_{\!\mu,\lambda}$, cf.\ \eqref{definition_konstante}. Let $\widetilde{L} \geq 1$ and $\lambda \leq E_0(\mu) + e_P(\mu)$ be arbitrary. It holds $\norm{a(\eta_\mu)} = \norm{a^*(\eta_\mu)} = \norm{\eta_\mu}$ and by 
\eqref{Estimate_Eta_mu},
\begin{align}
  \norm{\eta_\mu}^2 &= \frac{1}{L^2}\sum_{k^2 \leq \mu} 1 \leq \frac{\mu}{4\pi} + \frac{2 \sqrt{\mu}}{\pi L} + 
\frac{3}{L^2}. 
\label{Sum_over_one_from_k_to_mu}
\end{align}
Moreover, by \eqref{Pull_Through_Formula},
\begin{align*}
  A_{\mu,\lambda} A_{\mu,\lambda}^* &= \left( \sum_{k^2 > \mu} \frac{L^{-2}}{(T + k^2 - \lambda)^2} - \sum_{k^2, l^2 > \mu} \!\! 
a_l \: \frac{L^{-2}}{(T + k^2 + l^2 - 
\lambda)^2} \: a_k^* \right) \rst \HH_{\Nmu-1}.
 \end{align*}
We drop the second term employing a positivity argument similar to \eqref{Easy_Positivity_Argument}. Then, 
using $T \geq E_0(\mu) - \mu$ on $\HH_{\Nmu-1}$, we obtain
\begin{equation} \label{Bound_A_mu}
   A_{\mu,\lambda} A_{\mu,\lambda}^*
   \leq \frac{1}{L^2} \sum_{k^2 > \mu} \frac{1}{(E_0(\mu) - \mu + k^2 - \lambda)^2}.
\end{equation}
Next we invoke $E_0(\mu) - \lambda \ge - e_P(\mu)$, and then apply Lemma \ref{Riemann} (b), to bound \eqref{Bound_A_mu} by
$$
\frac{1}{L^2} \sum_{k^2 > \mu} \frac{1}{(k^2 - \mu -e_P(\mu) )^2}
\leq \frac{1}{4\pi |e_P(\mu)|} + \frac{1}{2 L |e_P(\mu)|^{3/2}} + \left( \frac{4 \sqrt{\mu}}{\pi L} + \frac{6}{L^2} \right) \frac{1}{|e_P(\mu)|^2}.
$$
In this estimate we also employed
$$
   \int\limits_{\sqrt{\mu}}^\infty \frac{dt}{(t^2 - \mu - e_P(\mu))^2}
   \leq \int\limits_{\sqrt{\mu}}^\infty \frac{dt}{((t-\sqrt{\mu})^2 - e_P(\mu))^2}
   = \frac{\pi}{4} \cdot \frac{1}{|e_P(\mu)|^{3/2}}.
$$
With the bounds for $\norm{a(\eta_\mu)}$ and $\norm{A_{\mu,\lambda}^*}$ we obtain
$$
   r_{\!\mu,\lambda}^2 \leq 4 \left( \frac{\widetilde{\mu}}{4\pi} + \frac{2 \sqrt{\widetilde{\mu}}}{\pi \widetilde{L}} + \frac{3}{\widetilde{L}^2} \right) 
\cdot \left( 
\frac{1}{4\pi \widetilde{z}_P} + \frac{1}{2 \widetilde{L} \widetilde{z}_P^{3/2}} + \left( \frac{4 \sqrt{\widetilde{\mu}}}{\pi \widetilde{L}} + 
\frac{6}{\widetilde{L}^2} \right) \frac{1}{\widetilde{z}_P^2} \right),
$$
where $\widetilde{z}_P := |e_P(\mu)|/|E_B|$. Thus, by Lemma \ref{Lemma_bound_for_e_p}, there is $C > 0$ such that if $\widetilde{\mu}$ is large,
\begin{equation} \label{Estimate_R_mu}
  r_{\!\mu,\lambda} \leq C \sqrt{\log\widetilde{\mu}}
\end{equation}
for all $\widetilde{L} \geq 1$ and $\lambda \leq E_0(\mu) + e_P(\mu)$. Together with \eqref{Constant_C_3_Inverse} this 
implies Statement (a).

For the proof of (b), we write \eqref{Perturbed_Polaron_Equation} as $E_0(\mu) - \lambda = f(\lambda)$ in terms of the continuous, strictly monotonically 
increasing function $f: (-\infty, E_0(\mu) + e_P(\mu)) \to \R$ with
$$
  f(\lambda) = \frac{1}{L^2} \sum_{k^2 \leq \mu} (G_\mu(E_0(\mu)-k^2-\lambda)-r_{\mu,\lambda})^{-1}.
$$
By the definition of $G_\mu$ and \eqref{Estimate_R_mu}, $f(\lambda) \to 0$ as $\lambda \to -\infty$. Next consider $f(\Lambda(\mu))$ with $\Lambda(\mu) := 
E_0(\mu) + e_P(\mu)$. With the help of \eqref{Polaron_Equation},
$$
   f(\Lambda(\mu)) = -e_P(\mu) + \frac{r_{\mu,\Lambda(\mu)}}{L^2} \sum_{k^2 \leq \mu} \frac{1}{(G_\mu(-k^2-e_P(\mu))-r_{\mu,\Lambda(\mu)})\cdot 
G_\mu(-k^2-e_P(\mu))},
$$
and by (a), $f(\Lambda(\mu)) \geq -e_P(\mu)$ if $\widetilde{\mu}$ is large enough. These observations imply that there is a unique $\lambda(\mu) \in (-\infty, 
\Lambda(\mu))$ satisfying $E_0(\mu) - \lambda(\mu) = f(\lambda(\mu))$.

To obtain the stated bound for $\lambda(\mu)$, apply $E_0(\mu) - \lambda(\mu) \geq -e_P(\mu)$ to the argument of $G_\mu$ in \eqref{Perturbed_Polaron_Equation} and 
combine the resulting estimate with \eqref{Polaron_Equation} to obtain
$$
   E_0(\mu) + e_P(\mu) - \lambda(\mu) \leq \frac{1}{L^2} \sum_{k^2 \leq \mu} \frac{r_{\!\mu,\lambda(\mu)}}{(G_\mu(-k^2-e_P(\mu)) - r_{\!\mu,\lambda(\mu)})\cdot 
G_\mu(-k^2-e_P(\mu))}.
$$
By $G_\mu(-k^2-e_P(\mu)) \geq G_\mu(-\mu-e_P(\mu))$ for $k^2 \leq \mu$ and \eqref{Polaron_Equation} we arrive at
$$
   E_0(\mu) + e_P(\mu) - \lambda(\mu) \leq |e_P(\mu)| \cdot \frac{r_{\!\mu,\lambda(\mu)}}{G_\mu(-\mu-e_P(\mu)) - r_{\!\mu,\lambda(\mu)}}.
$$
Statement (b) now follows from \eqref{Constant_C_3_Inverse} and \eqref{Estimate_R_mu}.
\end{proof}

\section*{Appendix A}
\label{Sect:Appendix_A}

\setcounter{thm}{0}
\renewcommand{\thesection}{A}

In this appendix we prove \eqref{Upper_Lower_Bound_Criterion} by applying Theorem 6.3 and Corollary 
6.4 of \cite{GL_Variational} to the present case.

The one-body operator $h$ is the limit in the norm resolvent sense of
$$
   h_n := -\Delta - g_n \ket{\eta_n} \bra{\eta_n}
$$ 
as $n \to \infty$, where $\eta_n := L^{-1} \cdot \sum_{k^2 \leq n} \varphi_k$ and $g_n^{-1} = \sprod{\eta_n}{(-\Delta - 
E_B)^{-1} 
\eta_n}$. Thus, $\exp(ith_n) \to \exp(ith)$ strongly as $n \to \infty$ for all $t \in \R$, cf.\ \cite[Theorem 
VIII.21]{RS1}. Hence, $(\exp(ith_n))^{\otimes\Nmu}$ converges strongly to $(\exp(ith))^{\otimes\Nmu}$ for all $t 
\in \R$, where the former is 
the unitary group of 
$$ 
   H_\mu^{(n)} := d\Gamma(h_n) \rst \HH_{\Nmu} = (T - g_n a^*(\eta_n) a(\eta_n))\rst \HH_{\Nmu}
$$
and the latter is the unitary group of $H_\mu = d\Gamma(h) \rst \HH_{\Nmu}$. This implies that $H_\mu^{(n)} \to 
H_\mu$ as $n \to \infty$ in the strong resolvent sense.

The resolvent of $H_\mu^{(n)}$ can be given explicitly. In fact,
\begin{equation} \label{Regularized_Resolvent}
   (H_\mu^{(n)} - z)^{-1} = (T - z)^{-1} + (T - z)^{-1} a^*(\eta_n) \: \phi_\mu^{(n)}(z)^{-1} \: a(\eta_n) (T - 
z)^{-1},
\end{equation}
on $\HH_{\Nmu}$ for all $z \in \rho(T \rst \HH_{\Nmu})$ satisfying $0 \in \rho(\phi_\mu^{(n)}(z))$, where
$$
   \phi_\mu^{(n)}(z) := (g_n^{-1} - a(\eta_n) (T - z)^{-1} a^*(\eta_n)) \rst \HH_{\Nmu - 1}.
$$
For $z < 0$ in normal-ordered form,
\begin{equation} \label{Def:Phi_n}
  \phi_\mu^{(n)}(z) = \frac{1}{L^2} \sum_{k^2 \leq n} \left( \frac{1}{k^2 - E_B} - \frac{1}{T + k^2 - z} \right) 
+ \frac{1}{L^2} \sum_{k^2,l^2 \leq n} a_k^* \, \frac{1}{T + k^2 + l^2 - z} \, a_l 
\end{equation}
on $\HH_{\Nmu}$.

With this at hand one shows the following facts about the terms in \eqref{Regularized_Resolvent}.
\begin{enumerate}[label=(\roman*)]
  \item $a(\eta_n) (T-z)^{-1}$ converges in $\LL(\HH_{\Nmu},\HH_{\Nmu-1})$ to an operator $R_z$. The existence of this 
limit is easily established with the help of \eqref{Pull_Through_Formula}.
  \item Let $D \subseteq \HH_{\Nmu-1}$ denote the space of finite linear combinations of states of the form 
$\varphi_{k_1} \wedge ... \wedge \varphi_{k_{\Nmu-1}}$. Then for $\psi \in D$ and $z < 0$, $$\phi_\mu^{(n)}(z) \psi \to 
\phi_\mu(z) \psi \qquad (n \to \infty)$$ (cf. \eqref{Def:Phi}, \eqref{Def:Phi_n}), and $\phi_\mu(z)$ 
is essentially self-adjoint on $D$. For the latter fact note that the last term in \eqref{Def:Phi} is a bounded 
operator.
  \item Since the last term in \eqref{Def:Phi_n} is a positive operator (cf. \eqref{Easy_Positivity_Argument}) and $T 
\rst \HH_{\Nmu-1} \geq E_0(\mu) - \mu$, there is a $c_z > 0$ for each $z < E_0(\mu) - \mu + E_B$ such that 
$\phi_\mu^{(n)}(z) \geq c_z$ for all $n \in \N$.
\end{enumerate}
These facts imply that we are in the situation of Theorem 4.2 in \cite{GL_Variational}. Hence, by this theorem
\begin{equation} \label{Many_Body_Resolvent}
   (H_\mu - z)^{-1} = (T - z)^{-1} + R_{\overline{z}}^* \phi_\mu(z)^{-1} R_z
\end{equation}
on $\HH_{\Nmu}$ for $z < E_0(\mu) - \mu + E_B$.

Since $(T - z)^{-1} \rst \HH_{\Nmu}$ is compact and $\ker R_z^* = \{ 0 \}$, we can apply the results of 
\cite[Section 6]{GL_Variational} including Theorem 6.3 and Corollary 6.4, which 
yields \eqref{Upper_Lower_Bound_Criterion}.

\section*{Appendix B}
\label{Sect:Appendix_B}

\setcounter{thm}{0}
\renewcommand{\thesection}{B}

\begin{lemma} \label{Riemann} (a) Let $f:[0,\infty)\to[0,\infty)$ be monotonically decreasing. Then,
\begin{align}
 \left| \frac{1}{L^2} \sum_{k}f(k^2) - \frac{1}{2\pi}\int\limits_0^{\infty} f(t^2)t\, dt \right| \leq \frac{2}{\pi L}\int\limits_0^{\infty}f(t^2)\,dt + \frac{3 
f(0)}{L^2}.\label{Lemma_Riemann_a}
\end{align}
(b) Let $m>0$ and $f:[m,\infty)\to[0,\infty)$ be monotonically decreasing. Then,
\begin{align}
 \left| \frac{1}{L^2} \sum_{k^2>m}f(k^2) - \frac{1}{2\pi}\int\limits_{\sqrt{m}}^{\infty} f(t^2)t\, dt \right| 
 \leq \frac{2}{\pi L} \int\limits_{\sqrt{m}}^{\infty} f(t^2)\, dt + \left( \frac{4 \sqrt{m}}{\pi L} + \frac{6}{L^2} \right) f(m). \label{Lemma_Riemann_b}
\end{align}
\end{lemma}

\begin{proof} (a) By the symmetry of the function $k\mapsto f(k^2)$,
$$
  \sum_{k\in \frac{2\pi}{L}\Z^2}f(k^2) = 4 \sum_{k\in \frac{2\pi}{L}\Z_{+}^2}f(k^2) + 4 \sum_{k\in \frac{2\pi}{L}(\Z_{+}\times\{0\})}f(k^2) + f(0),
$$
which, by the monotonicity of $f$ is bounded from above by
$$
   \frac{L^2}{2\pi}\int\limits_0^{\infty} f(t^2)t\, dt + \frac{2L}{\pi}\int\limits_0^{\infty}f(t^2)\,dt + f(0).
$$
Moreover, by the symmetry of the function $k\mapsto f(k^2)$,
$$
  \sum_{k\in \frac{2\pi}{L}\Z^2}f(k^2) = 4 \sum_{k\in \frac{2\pi}{L}\Z_{\geq 0}^2}f(k^2) - 4 \sum_{k\in \frac{2\pi}{L}(\Z_{+}\times\{0\})}f(k^2) - 3f(0),
$$
which, by the monotonicity of $f$ is bounded from below by
$$
  \frac{L^2}{2\pi}\int\limits_0^{\infty} f(t^2)t\, dt - \frac{2L}{\pi}\int\limits_0^{\infty}f(t^2)\,dt - 3 f(0).
$$
The combination of upper and lower bound yields the statement of the lemma.\medskip

\noindent (b) We write
\begin{align*}
\frac{1}{L^2}\sum_{k^2>m} f(k^2) = 
\frac{1}{L^2}\sum_{k} g_m(k^2) - \frac{1}{L^2}\sum_{k} h_m(k^2),
\end{align*}
with the non-negative, monotonically decreasing functions
\begin{align*}
g_m(k^2) & := \chi_{(m,\infty)}(k^2) f(k^2) + \chi_{[0,m]}(k^2) f(m),\hspace{1cm}
h_m(k^2) := \chi_{[0,m]}(k^2) f(m).
\end{align*}
The left hand side of \eqref{Lemma_Riemann_b} can thus be estimated in terms of
\begin{align*}
 \left| \frac{1}{L^2}  \sum_{k} g_m(k^2) - \frac{1}{2\pi} \int\limits_{0}^\infty g_m (t^2) t\; dt \right|
 + \left| \frac{1}{L^2}  \sum_{k} h_m(k^2) - \frac{1}{2\pi} \int\limits_{0}^\infty h_m (t^2) t\; dt \right| .
\end{align*}
To the latter expression we apply \eqref{Lemma_Riemann_a} leading to the stated bound in (b).
\end{proof}

\section*{Appendix C}
\label{Sect:Appendix_C}

\setcounter{thm}{0}
\renewcommand{\thesection}{C}

\begin{lemma} \label{Lemma_min_max_estimate} Let $A$ and $B$ be two self-adjoint operators on a Hilbert space $\HH$ such that $A - B$ is of rank one. Then,
$$
   \mu_{i+1}(A) \geq \mu_i(B)
$$
for all $i \in \N$. Here $\mu_i(A)$ and $\mu_i(B)$ denote the min-max values of $A$ and $B$, respectively.
\end{lemma}

\begin{proof}
Let $i \in \N$. By definition of the min-max values,
\begin{align*}
   \mu_{i+1}(A) = \sup_{\psi_1,...,\psi_i} \: \inf_{\varphi \perp \psi_1, ... ,\psi_i} \frac{\sprod{\varphi}{A \varphi}}{\norm{\varphi}^2}.
\end{align*}
Restricting the supremum in this expression by fixing $\psi_i$ to a non-zero $v \in \ran(A-B)$ yields a lower bound for $\mu_{i+1}(A)$. Since $A - B$ is of 
rank one, $0 \neq \varphi \perp v$ implies $v \in \ran(A-B)^\perp = \ker(A-B)$ and $\sprod{\varphi}{A \varphi} = 
\sprod{\varphi}{B \varphi}$. Hence,
\begin{align*}
   \mu_{i+1}(A) \geq \sup_{\psi_1,...,\psi_{i-1}} \: \inf_{\varphi \perp \psi_1, ... ,\psi_{i-1},v} \frac{\sprod{\varphi}{B \varphi}}{\norm{\varphi}^2}.
\end{align*}
Lowering the right hand side by extending the infimum to all $\varphi$ perpendicular to $\psi_1, ..., \psi_{i-1}$ only, yields the statement of the lemma.
\end{proof}

\begin{proof}[Proof of \eqref{Eigenvalue_estimates_rankone_perturbation}] We apply Lemma \ref{Lemma_min_max_estimate} to $A := -(h-z)^{-1}$ and $B := 
-(-\Delta-z)^{-1}$ for $z<  E_B$. Since the spectra of these operators are purely discrete and bounded from below, the min-max values of $A$ and $B$ coincide 
with the respective eigenvalues. Then, $\lambda_i = z - 1/\mu_i(A)$ and $\lambda_i^0 = z - 1/{\lambda_i(B)}$ for all $i \in \N$, and 
\eqref{Eigenvalue_estimates_rankone_perturbation} is a direct consequence of Lemma \ref{Lemma_min_max_estimate}.
\end{proof}

\bigskip\noindent
\textbf{Acknowledgements.} We thank Marcel Griesemer for extended discussions and helpful remarks. Our work was supported by the \emph{Deutsche 
Forschungsgemeinschaft (DFG)} through the Research Training Group 1838: \emph{Spectral Theory and Dynamics of Quantum Systems}.

\end{document}